%% file: 00-main.tex
\documentclass[a4paper,UKenglish,cleveref,thm-restate]{lipics-v2021}

\usepackage[subrefformat=simple,labelformat=simple,position=b]{subcaption}

\usepackage{graphicx}
\usepackage{amssymb}
\usepackage{amsmath}
\usepackage{complexity}
\usepackage{thm-restate}
\usepackage{cleveref}
\usepackage{xcolor}
\usepackage{xspace}
\usepackage{nicefrac}
\usepackage{enumerate}

\newcommand{\guardset}{\mathcal{G}}
\newcommand{\polygon}{\mathcal{P}}
\newcommand{\planarsat}{\textsc{Planar Monotone 3Sat}\xspace}
\newcommand{\dispersiveagp}{\textsc{Dispersive AGP}\xspace}

\hideLIPIcs

\graphicspath{{./figures/}}

\title{Dispersive Vertex Guarding for Simple and Non-Simple Polygons}

\titlerunning{Dispersive Vertex Guarding for Simple and Non-Simple Polygons}

\author{Sándor P. Fekete}{Department of Computer Science, TU Braunschweig, Braunschweig, Germany}{s.fekete@tu-bs.de}{https://orcid.org/0000-0002-9062-4241}{}
\author{Joseph S. B. Mitchell}{Department of Applied Mathematics and Statistics, Stony Brook University, Stony Brook, NY, USA}{joseph.mitchell@stonybrook.edu}{0000-0002-0152-2279}{}
\author{Christian Rieck}{Department of Computer Science, TU Braunschweig, Braunschweig, Germany}{rieck@ibr.cs.tu-bs.de}{https://orcid.org/0000-0003-0846-5163}{}
\author{Christian Scheffer}{Department of Electrical Engineering and Computer Science, Bochum University of Applied Sciences, Bochum, Germany}{christian.scheffer@hs-bochum.de}{https://orcid.org/0000-0002-3471-2706}{}
\author{Christiane Schmidt}{Department of Science and Technology, Link\"oping University,  Norrk\"oping, Sweden}{christiane.schmidt@liu.se}{https://orcid.org/0000-0003-2548-5756}{}

\authorrunning{S. P. Fekete, J. S. B. Mitchell, C. Rieck, C. Scheffer, and C. Schmidt}

\Copyright{Sándor P. Fekete, Joseph S. B. Mitchell, Christian Rieck, Christian Scheffer, Christiane Schmidt}



\funding{
	Work at TU Braunschweig is partially supported by the German Research Foundation~(DFG), project ``Computational Geometry: Solving Hard Optimization Problems'' (CG:SHOP), grant FE407/21-1.
	J.\,S.\,B.~Mitchell is partially supported by the National Science Foundation (CCF-2007275). 
	C.~Schmidt is partially supported by Vetenskapsr\r{a}det by the grant ``{\sc illuminate} provably good methods for guarding problems'' (2021-03810).
	}

\nolinenumbers

\EventEditors{}
\EventNoEds{1}
\EventLongTitle{}
\EventShortTitle{}
\EventAcronym{}
\EventYear{}
\EventDate{}
\EventLocation{}
\EventLogo{}
\SeriesVolume{}
\ArticleNo{}
\begin{document}

    \maketitle

    \begin{abstract}
        We study the \textsc{Dispersive Art Gallery Problem} with vertex guards:
        Given a polygon $\polygon$, with pairwise geodesic Euclidean vertex distance of at least $1$, and a rational number~$\ell$; 
        decide whether there is a set of vertex guards such that $\polygon$ is guarded, and the minimum geodesic Euclidean distance between any two guards (the so-called \emph{dispersion distance}) is at least~$\ell$.
        
        We show that it is \NP-complete to decide whether a~polygon with holes has a set of vertex guards with dispersion distance~$2$.
        On the other hand, we provide an algorithm that places vertex~guards in \emph{simple} polygons at dispersion distance at least~$2$.
        This result is tight, as there are simple polygons in which any vertex guard set has a dispersion distance of at most~$2$.
    \end{abstract}

	\input{01-introduction}
	\input{02-observations}

	\input{03-hardness}
	\input{04-worstcase-optimal}

	\input{05-conclusion}

    \bibliography{bibliography}
\end{document}

%% file: 01-introduction.tex
% !TeX root = 00-main.tex
\section{Introduction}\label{sec:introduction}
The \textsc{Art Gallery Problem} is one of the fundamental challenges in computational geometry.
It was first introduced by Klee in 1973 and can be stated as follows: 
Given a polygon~$\polygon$ with~$n$ vertices and an integer $k$; decide whether there is a set of at most $k$ many guards, such that these guards see all of $\polygon$, where a guard sees a point if the line segment connecting them is fully contained in the polygon.
Chv\'atal~\cite{chvatal1975combinatorial} and Fisk~\cite{fisk1978short} established tight worst-case bounds by showing that $\lfloor\nicefrac{n}{3}\rfloor$ many guards are sometimes necessary and always sufficient.
On the algorithmic side, Lee and Lin~\cite{LeeL86} proved \NP-hardness; more recently, 
Abrahamsen, Adamaszek, and Miltzow~\cite{agp-exist-r-complete} showed $\exists\mathbb{R}$-completeness, even for simple polygons.

In this paper, we investigate the \textsc{Dispersive AGP} in polygons with vertex guards:
Given a polygon $\polygon$ and a rational number $\ell$, find a set of vertex guards such that $\polygon$ is guarded and the minimum pairwise geodesic Euclidean distance between each pair of guards is at least $\ell$. 
(Note that the cardinality of the guard set does not come into play.)

\subsection{Our Contributions}

We give the following results for the \textsc{Dispersive Art Gallery Problem} in polygons with vertex guards. 

\begin{itemize} 
	\item For polygons with holes, we show \NP-completeness of deciding whether a pairwise geodesic Euclidean distance between any two guards of at least $2$ can be guaranteed. 
	\item For simple polygons, we provide an algorithm for computing a set of vertex guards of minimum pairwise geodesic distance of at least $2$.	
	\item We show that a dispersion distance of~$2$ is worst-case optimal for simple polygons.
\end{itemize}

\subsection{Previous Work}
Many variations of the classic \textsc{Art Gallery Problem} have been investigated~\cite{o1987art,shermer1992recent,Urrutia00}. 
This includes variants in which the number of guards does not play a role, such as the \textsc{Chromatic AGP}~\cite{erickson2010chromatic,FeketeFHM014,IwamotoI20} as well as the \textsc{Conflict-Free Chromatic AGP}~\cite{BartschiGMTW14,BartschiS14,hksvw-ccgoag-18}.

The \textsc{Dispersive AGP} was first introduced by Mitchell~\cite{joeDispersive}, and studied for the special case of polyominoes by Rieck and Scheffer~\cite{RieckS24}. 
They gave a method for computing worst-case optimal solutions with dispersion distance at least $3$ for simple polyominoes, and showed \NP-completeness of deciding whether a polyomino with holes allows a set of vertex guards with dispersion distance of $5$.

\subsection{Preliminaries}
Given a polygon $\polygon$ (possibly with holes), we say that two points $p,q\in \polygon$ \emph{see each other}, if the connecting line segment $\overline{pq}$ is fully contained in
$\polygon$.  A (finite) set of points $\guardset\subset \polygon$ is called a
\emph{guard set} for $\polygon$, if all points of~$\polygon$ are seen by
at least one point of $\guardset$. If $\guardset$ is a subset of the vertices of
$\polygon$, we are dealing with \emph{vertex guards}.

Distances between two points $p,q\in\polygon$ are measured
according to the Euclidean geodesic metric, i.e., that is the length of a shortest path between $p$ and $q$ that stays fully inside of~$\polygon$, and are denoted by $\delta(p,q)$.
The smallest distance between any two guards within a guard set is called its  \emph{dispersion distance}.

%% file: 02-observations.tex
% !TeX root = 00-main.tex
\section{First Observations}\label{sec:observations}

We start with two easy observations; the second resolves an open problem by Rieck and Scheffer~\cite{RieckS24}, who raised the question about the ratio of the cardinalities of guard sets in optimal solutions for the \textsc{Dispersive AGP} and the classical AGP.

\subsection{Shortest Polygon Edges Are Insufficient as Lower Bounds}
To see that an optimal dispersion distance may be considerably shorter than the shortest polygon edge, consider~\Cref{fig:shortestedgeexample}.
Every edge in the polygon has similar length (say, between~$1$ and $1+\varepsilon$). 
To guard the colored regions, one of each of the same colored vertices needs to be in the guard set.
This results in two guards that are arbitrarily close to each other. 

\begin{figure}[htb]
	\centering
	\includegraphics[scale=0.6]{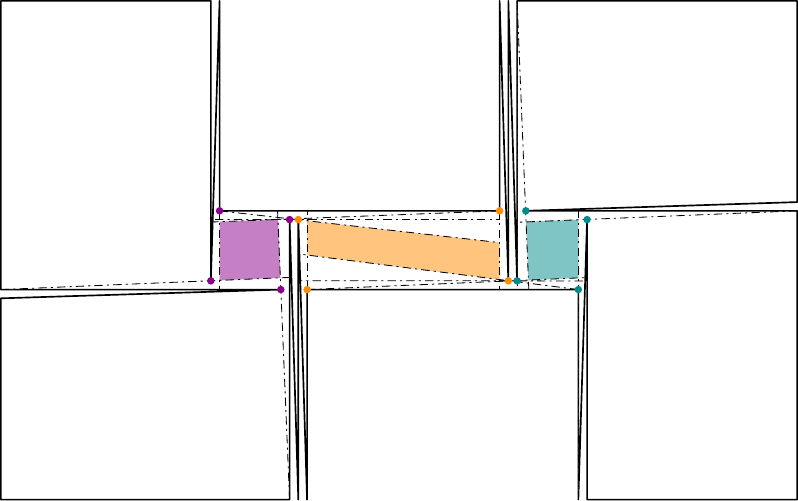}
	\caption{A polygon in which edges have similar length.}
	\label{fig:shortestedgeexample}
\end{figure}

This motivates our assumption that the geodesic distance between any pair of vertices is at least~$1$.

\subsection{Optimal Solutions May Contain Many Guards}
Even for a polygon that can be covered by a small number of guards, an optimal solution for the \textsc{Dispersive AGP} may contain arbitrarily many guards; see~\Cref{fig:dumbbell}. 
An optimal solution for the classical AGP consists of~$2$ guards placed at both ends of the central edge of length~$\varepsilon$.  
On the other hand, we can maximize the dispersion distance in a vertex guard set by placing one guard at the tip of each of the $\nicefrac{(n-2)}{2}$ spikes. 
These two sets have a dispersion distance of $\varepsilon$ and $2\zeta$, respectively, and the ratio $\nicefrac{2\zeta}{\varepsilon}$ can be arbitrarily large.

\begin{figure}[htb]
	\centering
	\includegraphics{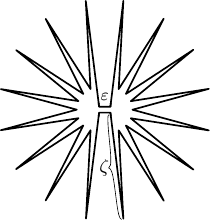}
	\caption{A polygon for which the optimal guard numbers for AGP and \textsc{Dispersive AGP} differ considerably.}
	\label{fig:dumbbell}
\end{figure}

%% file: 03-hardness.tex
% !TeX root = 00-main.tex
\section{\NP-Completeness for Polygons with Holes}\label{sec:hardness}

We now study the computational complexity of the \dispersiveagp for vertex guards in non-simple polygons.

\begin{restatable}{theorem}{hardnesswithholesTheorem}\label[theorem]{thm:dispersion-distance-2-np-hard}
	It is \NP-complete to decide whether a polygon with holes and geodesic vertex distance of at least $1$ allows a set of vertex guards with dispersion distance~$2$.
\end{restatable}

We first observe that the problem is in \NP. For a potential guard set $\guardset$,
we can check the geodesic distance between any pair of vertices $g_1,g_2\in \guardset$ as follows.
Because any two polygon vertices have mutual distance of at least $1$,
a shortest geodesic path between $g_1$ and $g_2$ consisting of at least
two edges has a length of at least~$2$. This leaves checking the
length of geodesic paths consisting of a single edge, which is straightforward.

\subsection{Overview and Gadgets}

For showing \NP-hardness, we utilize the \NP-complete problem \textsc{Planar Monotone 3Sat}~\cite{dbk-obspp-10},
which asks for the satisfiability of a Boolean 3-CNF formula, for which the literals in each 
clause are either all negated or all unnegated, and the corresponding variable-clause incidence graph is planar.

To this end, we construct gadgets to represent (i) variables, (ii) clauses, (iii) a gadget that splits the respective assignment, and (iv) gadgets that connect subpolygons while maintaining the given truth assignment.

\subparagraph*{Variable Gadget.} A \emph{variable gadget} is shown in~\Cref{fig:variablegadget}. 
Its four vertices $v_1, v_2, v_3, v_4$ are placed on the vertices of a rhombus (shown in grey) formed by two adjacent equilateral
triangles of side length~$1$. We add two sharp spikes by connecting two additional vertices
($v_5$ and $v_6$) to the two
pairs $v_1,v_3$ and $v_2,v_4$, respectively; the edges $\{v_3,v_5\},\{v_2,v_6\}$ have unit-length. (The function of these spikes is to impose an upper bound of 2 on
the achievable distance.) We also attach two narrow polygonal corridors to two other pairs of vertices, indicated
in green for the pair $v_1, v_2$, 
and in red for $v_3, v_4$.
These corridors have appropriate width, up to $1$, at the other end,
to attach them to other gadgets.

\begin{figure}[htb]
	\centering
	\includegraphics{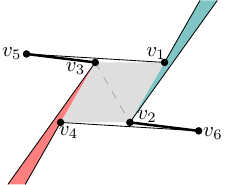}
	\caption{A variable gadget.}
	\label{fig:variablegadget}
\end{figure}

\begin{lemma}\label[lemma]{lem:variablegadget}
	Exactly four vertex guard sets realize a dispersion distance of $2$ to guard the subpolygon ${\mathcal{P}_v = (v_1, v_2, v_6, v_4, v_3, v_5, v_1)}$ of a variable gadget.
\end{lemma}

\begin{proof}
	As shown in \Cref{fig:variablegadget}, at least one guard has to be placed on one of $\{v_1, v_2, v_3, v_4\}$
to guard the rhombus that represents the variable.
	Conversely, it is easy to see that the dispersion distance is less than $2$ 
if more than one guard is chosen from $\{v_1, v_2, v_3, v_4\}$. Furthermore, if we choose
$v_1$ or $v_3$, the spike at $v_5$ is guarded, and we can choose $v_6$ (which has distance
$2$ from both $v_1$ and~$v_3$) to guard the other spike; conversely, a guard at $v_2$ or $v_4$ 
covers the spike at $v_6$ and allows a guard at $v_5$.
	
Now a guard from $\{v_1,v_2\}$ also covers the green portion of the 
polygon; this will correspond to setting the variable to \texttt{true}.
On the other hand, a guard from $\{v_3, v_4\}$ also
covers the red portion of the polygon, corresponding to setting 
the variable to \texttt{false}. 
\end{proof}

\subparagraph*{Clause Gadget.} A \emph{clause gadget} is depicted in~\Cref{fig:clausegadget}.
Its three vertices lie on the vertices of an equilateral triangle of side length~$1$; 
attached are narrow polygonal corridors, which are nearly
parallel to the triangle edges, each using two of the triangle vertices as end points.
These corridors have appropriate width, up to $1$, at the other end,
to attach them to other~gadgets.

\begin{figure}[htb]
	\centering
	\includegraphics{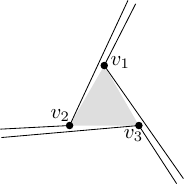}
	\caption{A clause gadget.}
	\label{fig:clausegadget}
\end{figure}

\bigskip
\begin{observation}
	As the vertices $\{v_1,v_2,v_3\}$ have a pairwise distance of $1$, only a single guard can be placed within a clause gadget, if the guard set have to realize a dispersion distance of at least~$2$.
	A direct consequence is that no more than two of the incident corridors can be guarded by a guard placed on these vertices; hence, at least one corridor needs to be seen from somewhere else, which in turn corresponds to satisfying the clause.
\end{observation}

\subparagraph*{Split Gadget.} A \emph{split gadget} is illustrated in~\Cref{fig:splitgadget}. 
It~has one incoming horizontal polygonal corridor, ending at two vertices ($v_1$ and $v_2$) within vertical distance~$1$. 
These vertices form an equilateral triangle with a third vertex, $v_4$, where the polygon splits into 
two further corridors, emanating horizontally from vertices $v_3, v_6$, and $v_5, v_7$, respectively. 
For the upper corridor, the vertices $v_1, v_3, v_4, v_6$ form slightly distorted adjacent equilateral unit triangles: We move $v_3$ and $v_6$ slightly upwards, such that the edges $\{v_1,v_3\}$ and $\{v_4,v_6\}$ as well as the distances between $v_1$ and~$v_4$ and between $v_3$ and $v_6$ remain $1$, but the distance between $v_3$ and $v_4$ increases to $1+\varepsilon$. 
An analogous construction yields the lower outgoing horizontal corridor. 
Both of these corridors start with a height smaller than $1$, but can end with a height of $1$ or a very small height.
\begin{figure}[htb]
	\centering
	\includegraphics{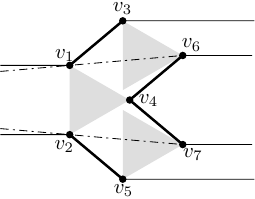}
	\caption{A split gadget.}
	\label{fig:splitgadget}
\end{figure}
\begin{lemma}\label[lemma]{lem:splitgadget}
	The split gadget correctly forwards the respective variable assignment.
\end{lemma}
\begin{proof}
	We refer to~\Cref{fig:splitgadget} and distinguish two cases.
	First, assume that the variable adjacent to the left is set to \texttt{true}, implying
	that the connecting corridor is already guarded.
	Therefore, two guards placed on $v_6$ and $v_7$ guard the whole subpolygon, and in particular both corridors to the right.

	Now assume that the variable is set to \texttt{false}, implying
	that the corridor to the left is not fully guarded yet. Because 
	this corridor is constructed long enough to contain the intersection of
	the (dotted) lines through $v_1, v_6$ and $v_2, v_7$, we need 
	to place a guard at one of the vertices in $\{v_1, v_2, v_4\}$.
	Then no further guard can be placed in a distance of at least~$2$, and
	the corridors to the right are not guarded, as claimed.
\end{proof}

\subparagraph*{Connector Gadget.} The \emph{connector gadget} is depicted in~\Cref{fig:connectorgadget}. 
The distance between all pairs of vertices is at least~$1$ and less than $2$. 
Furthermore, a guard on either $v_1$ or~$v_2$ cannot see the horizontal corridor, while a guard on $v_3$ or~$v_4$ does not see the vertical one.
\begin{figure}[htb]
	\centering
	\includegraphics{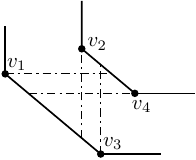}
	\caption{A connector gadget.}
	\label{fig:connectorgadget}
\end{figure}

\begin{observation}
	The gadget is designed such that only a single guard can be placed on its vertices while maintaining a distance of at least $2$.
	If a previously placed guard already sees the vertical corridor, we can place another one to see the horizontal corridor as well. 
	On~the other hand, no guard sees both corridors simultaneously.	
	Thus, we propagate a truth~assignment.
\end{observation}

\subsection{Construction and Proof}

We now describe the construction of the polygon for the reduction, and complete the proof.

\hardnesswithholesTheorem*

\begin{proof}
	To show \NP-hardness, we reduce from \planarsat. 
	For any given Boolean formula $\varphi$, we construct a polygon $\polygon_\varphi$ as an instance of \dispersiveagp as follows.
	Consider a planar embedding of the variable-clause incidence graph of $\varphi$, place the variable gadgets in a row, and clause gadgets that only consist of unnegated literals or entirely of negated literals to the top or to the bottom of that row, respectively, as illustrated in~\Cref{fig:planarsatembedding}.
	Furthermore, connect variables to clauses via a couple of connector gadgets, and introduce split gadgets where necessary.
	
	\begin{figure}[htb]
		\centering
		\includegraphics{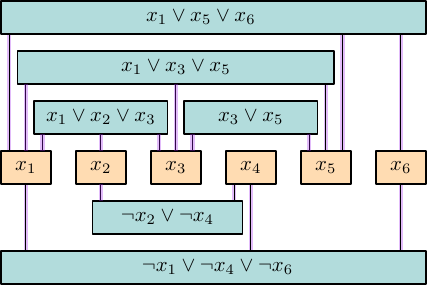}
		\caption{Rectilinear embedding of a \planarsat instance.}
		\label{fig:planarsatembedding}
	\end{figure}
	
	\pagebreak
	\begin{claim}
		If $\varphi$ is satisfiable, then $\polygon_\varphi$ has a vertex guard set with dispersion distance $2$.
	\end{claim}
	
	\begin{claimproof}
		Given a satisfying assignment, we construct a set of vertex guards with a dispersion distance of $2$: 
		For~every variable that is set to \texttt{true}, we place guards on~$\{v_1, v_6\}$,  and for every variable that is set to \texttt{false}, we place guards on $\{v_4, v_5\}$ within the respective variable gadget.
		Furthermore, we place guards for split and connector gadgets to maintain the given assignments.
		As we have a satisfying assignment, each clause is satisfied by at least one literal, i.e., at least one corridor incident to the clause gadget is already guarded. 
		Therefore, we can place one guard in each clause gadget. 
		This yields a guard set with a dispersion distance of $2$.
	\end{claimproof}
	
	\begin{claim}
		If $\polygon_\varphi$ has a vertex guard set with dispersion distance~$2$, then $\varphi$ is satisfiable.
	\end{claim}
	
	\begin{claimproof}
		As we have a set of vertex guards with a dispersion distance of~$2$, there is only a single guard placed within each clause gadget. 
		Furthermore, no guard set can have larger dispersion distance within a variable gadget.
		As~argued before, there is no guard set with a dispersion distance larger than $2$ in the split and connector gadgets.
		Therefore, the vertex guards placed within the variable gadgets provide a suitable variable assignment for~$\varphi$.
	\end{claimproof}
	
	\medskip
	Given that the problem is in \NP, these two claims complete the proof.
\end{proof}

%% file: 04-worstcase-optimal.tex
\section{Worst-Case Optimality for Simple Polygons}\label{sec:worstcase-optimal}

In this section, we prove that a guard set realizing a dispersion distance of $2$ is worst-case optimal for simple polygons. 
In particular, we describe an algorithm that constructs such guard sets for any simple polygon.

First, we observe that there are polygons for which there is no guard set with a larger dispersion distance.

\begin{observation}\label[Observation]{lem:distance-two-necessary-simple}
	There are simple polygons with geodesic vertex distance at least $1$ for which every guard set has a dispersion distance of at most $2$.
\end{observation}

	Refer to~\Cref{fig:distance-two-necessary-simple}.
	Bold edges have length~$1$. One of the three vertices (with pairwise distance~$1$) 
	incident to the gray triangle $\Delta$ must be picked to guard $\Delta$,
	so no guard set can have a dispersion distance larger that~$2$.
	
	\begin{figure}[htb]
		\centering
		\includegraphics{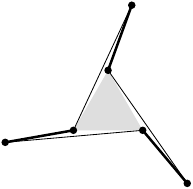}
		\caption{Godfried's favorite polygon.}
		\label{fig:distance-two-necessary-simple}
	\end{figure}
	
	From this, we can easily obtain polygons with \emph{any} number of vertices of
	dispersion distance at most $2$: Simply modify the polygon at the end of each spike.

In the remainder of this section, we provide a polynomial-time algorithm that constructs a guard set with dispersion distance of at least $2$.

We start with a useful lemma that provides some structural properties for the analysis.
Refer to \Cref{fig:auxiliary-layout} for visual reference.

\begin{lemma}\label[lemma]{le:for-contradiction}
Let ${\polygon= (v_1, v_2, \dots, v_7)}$ be a simple polygon with seven  vertices labeled in counterclockwise order.
Assume that the pairwise (geodesic) distance between all
pairs of vertices is at least $1$ and further that the following
properties are satisfied:
\smallskip
\begin{enumerate}
\item The distance between $v_1$ and $v_5$ is $\delta(v_1,v_5)<2$.
\item $v_2, v_4$, and $v_7$ are reflex, i.e., the interior angle at these vertices is strictly larger than~$180^\circ$.
\end{enumerate}
\smallskip
Then the geodesic distance between the two vertices $v_3$ and $v_6$ is $\delta(v_3,v_6)\geq 2$.
\end{lemma}

\begin{figure}[htb]
	\centering
	\includegraphics{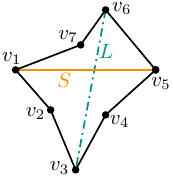}
	\caption{Schematic layout of $\polygon$.}
	\label{fig:auxiliary-layout}
\end{figure}

\begin{proof}
Throughout the proof, we will frequently make use of the assumption that $\delta(v_i,v_j)\geq 1$ for any $i\neq j$. 

As a first step, we argue that $v_1$ and $v_5$ must be mutually visible (along line segment $S$), as shown
in ~\Cref{fig:auxiliary-layout}: Otherwise, a shortest geodesic path
from $v_1$ to $v_5$ must visit one of the reflex vertices $v_2$, $v_4$, or $v_7$, implying
the contradiction $\delta(v_1,v_5)\geq 2$. 

By a similar argument, we claim that $v_3$ and $v_6$ are mutually visible (say, along segment~$L$); otherwise we can conclude that
$\delta(v_3,v_6)\geq 2$, and we are done. 

In the following, we prove that $L$
has length at least~$2$, by establishing the following two auxiliary claims.

\medskip
\descriptionlabel{(a)} The geodesic distance from $v_6$ to $S$ is at least $\nicefrac{\sqrt{3}}{2}$.

\descriptionlabel{(b)} The geodesic distance from $v_3$ to $S$ is at least ${2-\nicefrac{\sqrt{3}}{2}=1.13397\ldots}$.

\begin{figure}[htb]
	\centering
	\begin{subfigure}{0.5\textwidth}
		\centering
		\includegraphics{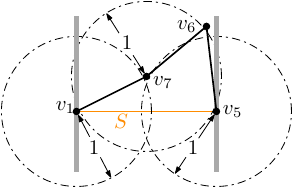}
		\caption{}
		\label{fig:auxiliary-lemma-b}
	\end{subfigure}\hfill
	\begin{subfigure}{0.5\textwidth}%
		\centering
		\includegraphics{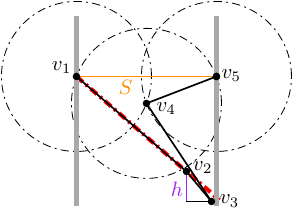}
		\caption{}
		\label{fig:auxiliary-lemma-c}
	\end{subfigure}\hfill
	\caption{Subpolygons for the proof of the two auxiliary claims: (a) The quadrangle $\polygon'=(v_1,v_5,v_6,v_7)$. (b)
The pentagon $\polygon''=(v_1,v_2,v_3,v_4,v_5)$.}
	\label{fig:auxiliary-lemma}
\end{figure}

To this end, assume that the cord $S$ lies horizontally,
with $v_1=(0,0)$ and $v_5=(x_5,0)$,  and partitions $\polygon$ into
two subpolygons, as illustrated in~\cref{fig:auxiliary-lemma}:
(a) the quadrangle ${\polygon':=(v_1,v_5,v_6,v_7)}$ above $S$,
and (b) the pentagon $\polygon'':=(v_1,v_2,v_3,v_4,v_5)$ below $S$. Because~$v_7$ is reflex,
it must lie inside the convex hull of $\polygon'$, which is spanned by the three
remaining vertices $v_1, v_5, v_6$. Analogously, $v_2$ and $v_4$ must lie inside
the convex hull of $\polygon''$, which is spanned by the three remaining vertices
$v_1, v_3, v_5$.

\smallskip
\descriptionlabel{For the auxiliary claim (a)}, refer to~\Cref{fig:auxiliary-lemma-b}. If $v_6=(x_6,y_6)$ 
lies outside the vertical strip defined by $0\leq x\leq x_5$, 
then its closest point on
$S$ is $v_5$ (for $x_6>x_5$) or a~point $q$ for which the geodesic to $S$ runs via $v_7$ (for $x_6<0$), so the minimum distance of~$v_6$ is~$\delta(v_1,v_6)\geq 1$, (or $\delta(v_5,v_6)\geq 1$, respectively). Therefore, the
convex hull of $\polygon'$ must lie within the strip, including $v_7$. 
Furthermore, $v_6$ must have the largest vertical distance from~$S$, so $v_7$ must lie within
the axis-aligned rectangle $R'$ of height $\nicefrac{\sqrt{3}}{2}$ above~$S$.
Consider the three circles $C_1$, $C_5$, and $C_7$ of unit radius around 
$v_1$, $v_5$, and~$v_7$.  It is straightforward to verify that $R'$ is
completely covered by $C_1$, $C_5$, and $C_7$, implying that $v_6$ cannot
lie inside~$R'$, and the first claim follows.

\smallskip
\descriptionlabel{For the auxiliary claim (b)}, refer to~\Cref{fig:auxiliary-lemma-c}.
Without loss of generality, assume that the vertical distance $-y_2$ of $v_2$ from $S$ is not smaller than the vertical distance
$-y_4$ of~$v_4$. Consider the horizontal positions $x_2, x_3, x_4$ of $v_2, v_3, v_4$. 
Because $v_2$ lies inside the convex hull of $\polygon''$, 
the assumption $x_2<0$ (which differs from the figure) implies that $x_3\leq x_2<0$; 
then the shortest distance from $v_2$ to $S$ is $\delta(v_2,v_1)\geq 1$,
and (because $v_2$ is reflex), a~shortest geodesic path from $v_3$ to $S$ 
passes through $v_2$, so $\delta(v_3,v_1)=\delta(v_3,v_2)+\delta(v_2,v_1)\geq 2$. Analogously, we can assume that 
$x_4\leq x_5$. Furthermore, the assumption on the 
relative vertical positions of $v_2$ and~$v_4$ implies that $v_4$
must also lie to the right of $v_1$, i.e., $x_4\geq 0$.

\smallskip
Consider $x_2\geq s$ and refer to~\Cref{fig:outlier}. 
Then the assumption $-y_2\leq \nicefrac{\sqrt{3}}{2}$ (together with 
$\delta(v_2,v_5)\geq 1$) implies that $x_2\geq x_5+\nicefrac{1}{2}$. Furthermore,
$v_4$ must lie above the edge $(v_1,v_2)$. 

\begin{figure}[htb]
	\centering
	\includegraphics{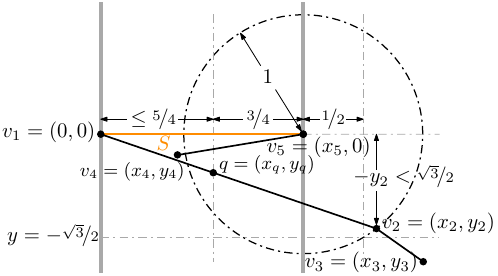}
	\caption{Estimating $\delta(v_2,v_4)$.}
	\label{fig:outlier}
\end{figure}

We now consider the point $q=(x_5-\nicefrac{3}{4}, y_q)$.
As $x_5< 2$, we conclude ${-y_q\leq \nicefrac{-y_2}{2}\leq  \nicefrac{\sqrt{3}}{4}}$.
Therefore, $\delta(q,v_5)\leq\sqrt{(\frac{4}{3})^2+(\frac{\sqrt{3}}{4})^2}=0.96824\ldots<1$.
Because $v_4$ must lie outside of the circle with radius~$1$ around~$v_5$, we conclude that
$x_4<x_5-\nicefrac{3}{4}$, implying that ${\delta(v_2,v_4)\geq \nicefrac{5}{4}=1.25}$.
Furthermore, $v_2$ cannot lie on the convex hull of $\polygon''$, thus, 
$x_3>x_2$ and $y_3<y_2$, implying ${\delta(v_3,v_4)>\delta(v_2,v_4)}$.
As the geodesically shortest path from $v_3$ to $S$ passes through~$v_4$, we conclude that the length of this path, $\delta(v_3,v_4)-y_4\geq \delta(v_3,v_4) \geq \delta(v_2,v_4)$
is bounded from below by $1.25$. Thus, we can assume that $-x_2\geq \nicefrac{\sqrt{3}}{2}$
in this case.

Alternatively, consider $x_2\leq x_5$. Then an argument for $v_1,v_2,v_4,v_5$
analogous to the one from claim (a) for $v_1,v_5,v_6,v_7$
also implies that $-x_2\geq \nicefrac{\sqrt{3}}{2}$.

\medskip
Consider the vertical distance $h:=-y_3+y_2$ between $v_3$ and $v_2$, and refer to~\Cref{fig:angles}.

\begin{figure}[htb]
	\centering
	\includegraphics{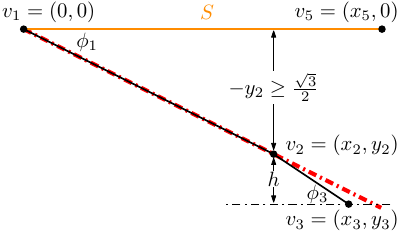}
	\caption{Angles and vertical distances at $v_1$ and $v_3$.}
	\label{fig:angles}
\end{figure}
To this end, note that the angle $\phi_1$ at $v_1$ between $(v_1,v_2)$ and
$S$ satisfies $\tan \phi_1=\nicefrac{-y_2}{x_2}\geq \nicefrac{\sqrt{3}}{4}= \eta$ because of 
$-y_2\geq \nicefrac{\sqrt{3}}{2}$ and $x_2\leq x_5 <2$. Because $v_2$ is reflex, the angle~$\phi_3$ between $(v_3,v_2)$ and a horizontal line at~$v_3$ satisfies
$\phi_3>\phi_1$; moreover, $\sin \phi_3=\frac{h}{\delta(v_2,v_3)}$,
with $\delta(v_2,v_3)\geq 1$, so $h\geq \sin\arctan \eta = \frac{\eta}{\sqrt{1+\eta^2}}=0.3973\ldots$.

This implies that the vertical distance $-y_3$ of $v_3$ to~$S$ (and thus
the distance of $v_3$ to $S$) is at least $\nicefrac{\sqrt{3}}{2}+0.3973=1.26338\ldots>1.13397\ldots=2-\nicefrac{\sqrt{3}}{2}$,
as~claimed.
\end{proof}

We now show the main result of this section.

\begin{theorem}\label[theorem]{thm:distance-two-sufficient-simple}
	For every simple polygon $\polygon$ with pairwise geodesic distance between 
 vertices at least $1$, there exists a guard set that has dispersion distance at least $2$.
\end{theorem}

\begin{proof}
Refer to~\Cref{fig:caterpillars,fig:guard-cases} for visual orientation.
By triangulating $\polygon$, we obtain a triangulation~$T$ whose dual graph is a tree~$T'$. 
We consider a path~$\Pi$ between two leaves (say, $t_1$ and $t_k$) in $T'$, and obtain a \emph{caterpillar}~$C'$
by adding as \emph{feet} all vertices adjacent to~$\Pi$; let~$C$ be the
corresponding set of triangles (shown in dark cyan in~\Cref{fig:caterpillars}).

Now the idea is to place guards on vertices of $C$ (that is a subset of the vertices of $\polygon$), aiming to see
all of~$C$. We then consider a recursive subdivision of~$\polygon$ into caterpillars,
by proceeding from foot triangles of covered caterpillars to ears, 
until all of $\polygon$ is covered; this corresponds to the colored subdivision in~\Cref{fig:caterpillars}. 

\begin{figure}[htb]
	\centering
	\includegraphics[page=1]{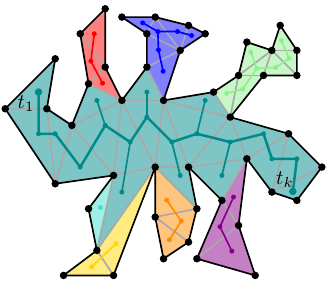}
	\caption{
		Polygon $\polygon$ in black, triangulation $T$ in gray, and a partition into (colored) caterpillars.
	}
	\label{fig:caterpillars}
\end{figure}

To cover $C$, we start by placing a guard on a vertex~$v_0$ of an ear triangle (say, $t_1$). If~$C'$ is a path (i.e., a~caterpillar without foot triangles), we can proceed in a straightforward manner:
Either the next triangles on the path are visible from the guard on $v_0$,
or there is a reflex vertex $v_r$ obstructing the view to a triangle $t_i$. 
In the latter case, we can
place the next guard on an unseen vertex $v_j$ of $t_i$, i.e., $v_j$ is not seen by any of the previously placed guards; by assumption, the distance
of $v_0$ and $v_r$ is at least~$1$, as is the distance of $v_r$ and $v_j$.
Because $v_r$ is reflex, a shortest path from $v_0$ to $v_j$ has length
at least~$2$ by triangle inequality.  

This leaves the case in which we have foot triangles, which is analyzed in the following.
Assume that we already placed a guard on a vertex incident to the path of the caterpillar.  
We argue how we proceed even if all path triangles have incident foot triangles, that is, we show that we can place a set of guards that together monitor all caterpillar triangles, while ensuring a distance of at least~$2$ between any pair of guards.  
Furthermore, whenever we place a guard in a foot triangle, then this guard is never needed to cover any path triangles,  hence,  even if not all path triangles have incident foot triangles, we yield a feasible guard placement.

\begin{figure}[htb]
	\centering
	\includegraphics[page=2]{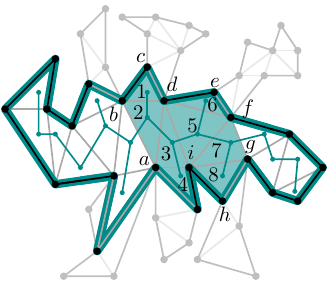}
	\caption{Vertices and triangular faces of a caterpillar for the proof of  \cref{thm:distance-two-sufficient-simple}.}
	\label{fig:guard-cases}
\end{figure}

In the recursive call, we also take into account what previously placed guards see; note that unseen vertices are feasible guard locations with a distance of at least $2$ to all previously placed guards.

To indicate that a vertex $v$ is reflex in the polygonal chain $u,v,w$, we say that \emph{$v$ is reflex w.r.t. $u-w$}; 
note that the polygonal chain $u,v,w$ must not be the polygon boundary. 
The line segment $uv$ contained in $\polygon$ is denoted by $\overline{uv}$; it is either a diagonal or a polygon edge.

Now we consider the situation in~\Cref{fig:guard-cases} and assume that a guard on $a$ has been placed to monitor the triangle to the left of $\overline{ab}$. 
We aim to monitor triangles $1,2,\ldots,8$.  The guard on $a$ sees the triangles
$2$, $3$ and $4$. If $a$ sees $c$, then $a$ sees triangle $1$ as well. If $a$ does not see~$c$, we place a guard on $c$ (in this case either $b$ or~$d$ is reflex w.r.t. $a-c$, thus, $c$ has distance at least~$2$ to $a$).  

\medskip
We now provide a case distinction on the next placement(s) of
guards. In case we placed a guard on $c$ in addition to the guard on $a$,
whenever we consider $a$ seeing vertices, this also includes vertex $c$.

\begin{description}
	\item[1.] If $i$ is reflex w.r.t. $a-g$, we place a guard on $g$; these guards see triangles $5,7$, and~$8$.
	\begin{description}
		\item[(a)] If $a$ or $g$ see $e$, then they also see triangle $6$.
		\item[(b)] Otherwise, we place another guard on $e$ (which has distance of at least $2$ to all guards placed before), which then monitors triangle $6$.
	\end{description}
	\medskip
	\item[2.] Otherwise, i.e., $i$ is not reflex w.r.t. $a-g$:
	\begin{description}
		\item[(a)] If $d$ is reflex w.r.t. $a-e$:
		\begin{enumerate}[{\bf i.}]
			\item If $\delta(a,f)\geq 2$, we place a guard on $f$ to see triangles $5,6,$ and $7$. If $h$ is seen by $a$ or~$f$, then also triangle $8$ is seen. Otherwise, we place a guard on~$h$ to see triangle $8$.
			\item Else if $\delta(a,g)\geq 2$, we place a guard on $g$, then $a$ and $g$ also see triangles $5,7,$ and $8$. If $e$ is seen by $a$ or $g$, then also triangle~$6$ is seen. Otherwise, we place a guard on~$e$ to see triangle $6$.
			\smallskip
			\item[]\hspace*{-1.5cm} {\em In the remaining cases iii.-- vi., we have $\delta(a,f)< 2$, $\delta(a,g)< 2$, thus, $a$ sees both $f$ and $g$.}
			\smallskip
			\item Else if $e$ does not see either $g$ or $h$ (which implies $\delta(e,h)>2$, $\delta(e,g)>2$): \\
			If $a$ does not see $h$, we place a guard on~$h$, which covers triangle $8$. Moreover, we also place a guard on $e$ (which is neither seen from $a$ or $h$), and the guards then also cover triangles $5,6,$ and $7$. 
			Otherwise, i.e., $a$ sees $h$, we place a guard on~$e$ to guarantee that triangles $5,6,7$, and $8$ are seen.
			\item Else if $e$ sees $g$, but does not see $h$:\\
			If $a$ does not see $h$, we place two guards on $e$ and $h$,  the guards together then guard triangles $5,6,7$, and $8$.
			Otherwise, $a$ sees $d,f,g,h,i$ and with that also triangles $5,7$, and $8$; we place a guard on $e$, which sees triangle $6$.
			\item Else if $e$ sees $h$, but does not see $g$:\\
			If $\delta(e,h)>2$, we place a guard on each $e$ and $h$, and thereby cover triangles $5,6,7$, and $8$. 
			If $\delta(e,h)<2$, \Cref{le:for-contradiction} yields a contradiction to $\delta(a,g)<2$ with $v_1=e$, ${v_2=d}$, $v_3=a$, $v_4=i$, $v_5=h$, $v_6=g$, and $v_7=f$.
			\item Else if $e$ sees $g$ and $h$:\\
			If $a$ sees $h$, we place a guard on $e$,  and the guards then cover triangles $5,6,7$, and~$8$. 
			Otherwise, we place a guard on~$h$, and if $h$ sees $f$, triangles $5,\ldots,8$ are seen. If not, we place a guard on~$e$ if $\delta(e,h)>2$ and cover triangles $5,\ldots,8$; otherwise, \Cref{le:for-contradiction} yields a contradiction to $\delta(e,h)<2$ with $v_1=a, v_2=i$, $v_3=h, v_4=g,v_5=f, v_6=e$, and $v_7=d$.
		\end{enumerate}
		\item[(b)] Otherwise, $a$ also sees $f$, hence,  triangles $5,6,$ and $7$ are covered.
		\begin{description}
			\item[i.] If $a$ sees $h$, it also sees triangle $8$.
			\item[ii.] If $a$ does not see $h$, we place a guard on $h$, which then sees triangle $8$.
		\end{description}
	\end{description}
\end{description}
The guards we place in foot triangles are never needed to cover path triangles, hence,  if some of the foot triangles did not exist, we can simply proceed along the caterpillar path (and place a guard there if a triangle is not (completely) seen).
\end{proof}

%% file: 05-conclusion.tex
\section{Conclusions and Future Work}

We considered the \textsc{Dispersive Art Gallery Problem} with vertex guards, both in simple polygons and in polygons with holes, where we measure distance in terms of geodesics between any two vertices.
We established \NP-completeness of the problem of deciding whether there exists a vertex guard set with a dispersion distance of $2$ for polygons with holes. 
For simple polygons, we presented a method for placing vertex guards with dispersion distance of at least $2$. 
While we do not show \NP-completeness of the problem in simple polygons, we conjecture the following.
\begin{conjecture}
	For a sufficiently large dispersion distance $\ell>2$, it is \NP-complete to decide whether a simple polygon
allows a set of vertex guards with a dispersion distance of at~least~$\ell$.
\end{conjecture}

Another open problem is to construct constant-factor approximation algorithms.
This hinges on good lower bounds for the optimum.

Both our work and the paper by Rieck and Scheffer~\cite{RieckS24} consider vertex guards.
This leaves the problem for point guards (with positions not necessarily at polygon vertices)
wide open. Given that the classical AGP for point guards is $\exists\mathbb{R}$-complete~\cite{agp-exist-r-complete},
these may be significantly more difficult to resolve.